\documentclass[twocolumn,amssymb, nobibnotes, aps, prl, superscriptaddress]{revtex4-2}

\usepackage{hyperref}
\usepackage{amsmath}
\usepackage{braket}
\usepackage{bm}
\usepackage{graphicx}

\usepackage{float} 
\makeatletter
\let\newfloat\newfloat@ltx
\makeatother

\usepackage{subfigure} 
\usepackage{color}
\usepackage[normalem]{ulem}
\usepackage{enumitem}
\usepackage{comment}
\usepackage{tabularx}
\usepackage{cancel}

\usepackage{algorithm}
\usepackage{algorithmic}

\usepackage{amsthm}          
\theoremstyle{definition}
\newtheorem{definition}{Definition}

\usepackage{amssymb}

\newtheorem{lemma}{Lemma} 

\begin{document}

\title{Public-Key Quantum Authentication and Digital Signature Schemes Based on the QMA-Complete Problem}
\author{Le-Ran Liu}
\affiliation{Department of Physics and HK Institute of Quantum Science \& Technology, The University of Hong Kong, Pokfulam Road, Hong Kong, China}
\affiliation{Hong Kong Branch for Quantum Science Center of Guangdong-Hong Kong-Macau Greater Bay Area, Shenzhen, China}
\author{Min-Quan He}
\affiliation{Department of Physics and HK Institute of Quantum Science \& Technology, The University of Hong Kong, Pokfulam Road, Hong Kong, China}
\affiliation{Hong Kong Branch for Quantum Science Center of Guangdong-Hong Kong-Macau Greater Bay Area, Shenzhen, China}
\author{Dan-Bo Zhang}
\affiliation{Key Laboratory of Atomic and Subatomic Structure and Quantum Control (Ministry of Education),\\ Guangdong Basic Research Center of Excellence for Structure and Fundamental Interactions of Matter,\\ and School of Physics, South China Normal University, Guangzhou 510006, China}
\author{Z. D. Wang}
\affiliation{Department of Physics and HK Institute of Quantum Science \& Technology, The University of Hong Kong, Pokfulam Road, Hong Kong, China}
\affiliation{Hong Kong Branch for Quantum Science Center of Guangdong-Hong Kong-Macau Greater Bay Area, Shenzhen, China}

\begin{abstract}
   We propose a quantum authentication and digital signature protocol whose security is founded on the Quantum Merlin Arthur~(QMA)-completeness of the consistency of local density matrices. The protocol functions as a true public-key cryptography system, where the public key is a set of local density matrices generated from the private key, a global quantum state. This construction uniquely eliminates the need for trusted third parties, pre-shared secrets, or authenticated classical channels for public key distribution, making a significant departure from symmetric protocols like quantum key distribution. We provide a rigorous security analysis, proving the scheme's unforgeability against adaptive chosen-message attacks by quantum adversaries. The proof proceeds by a formal reduction, demonstrating that a successful forgery would imply an efficient quantum algorithm for the QMA-complete Consistency of Quantum Marginal Problem~(QMP). We further analyze the efficiency of verification using partial quantum state tomography, establishing the protocol's theoretical robustness and outlining a path towards practical implementation.
\end{abstract}


\maketitle

\section{Introduction}
\label{sec:introduction}

The advent of scalable quantum computers threatens the security of conventional public-key cryptosystems~\cite{Shor1994, Mosca2018}. Shor’s polynomial-time algorithm for integer factorization and discrete logarithms undermines RSA~\cite{rsa1978} and elliptic-curve cryptography~\cite{miller1985, koblitz1987}, while Grover’s quadratic speed-up lowers the effective strength of symmetric ciphers~\cite{Grover1996}. Although block ciphers such as AES can be hardened by increasing the key size, the prospective collapse of public-key infrastructure compels the search for alternative paradigms~\cite{Bernstein2017}. 

In response to the quantum threat, two principal directions have emerged. The first is post-quantum cryptography (PQC)~\cite{bernstein2017post}, which seeks to develop classical cryptographic schemes believed to be secure against both classical and quantum attacks. Notable examples include lattice-based encryption~\cite{micciancio2011lattice}, code-based cryptography~\cite{overbeck2009code}, and multivariate polynomial schemes~\cite{sakumoto2011public}. While these constructions currently resist known quantum algorithms, their long-term security remains an open question, especially in the face of unforeseen advances in quantum algorithms or cryptanalysis. The second direction is quantum cryptography~\cite{gisin2002, portmann2022security}, which leverages the fundamental principles of quantum mechanics to achieve information-theoretic security. Protocols such as quantum key distribution~(QKD) offer provable security guarantees based on the laws of physics, rather than computational assumptions, and represent a fundamentally different approach to secure communication in the quantum era~\cite{Bennett1984, Scarani2009, renner2008, shor2000simple}.
QKD is a foundational quantum cryptography method that enables information-theoretically secure key exchange via quantum principles, e.g. the no-cloning theorem~\cite{wootters1982,gisin2002}. When combined with one-time pads (OTP), it guarantees unconditional security~\cite{Bennett1984,vernam1926}. However, QKD is fundamentally a system for generating symmetric keys that cannot function securely without a pre-authenticated classical channel~\cite{mayers1998,renner2008}. QKD also does not provide the public-private key pairs required by modern asymmetric cryptographic systems, which offers features like nonrepudiation and scalable trust models through infrastructures like Public Key Infrastructure (PKI)~\cite{diffie1976}.

Following QKD, another major branch of quantum communication is Quantum Secure Direct Communication (QSDC)~\cite{bostrom2002, deng2004, long2007}, which aims to transmit secret messages directly over a quantum channel without first establishing a secret key.  This approach promises enhanced efficiency and immediacy by condensing key distribution and ciphertext transmission into a single quantum process~\cite{long2007}.  Recent advances have demonstrated its potential, with experimental systems achieving communication over 100~km of fiber~\cite{Zhang2017, zhang2022} and the development of small multi-user networks~\cite{qi2021fifteen}.  However, QSDC is not entirely self-sufficient, as it still fundamentally requires an authenticated classical channel for coordination and eavesdropping detection~\cite{deng2004}.

While QKD and QSDC aim to secure the transmission process,
quantum identity authentication (QIA) and quantum digital signature (QDS) protocols are designed to leverage quantum infrastructure to achieve secure communication~\cite{barnum2002,gottesman2001,wallden2015}. Identity authentication is the process of ensuring the identity of the communicating parties, guaranteeing they are who they claim to be~\cite{zhang2006}. Digital signatures, on the other hand, are designed to ensure the authenticity and integrity of the message itself, providing guarantees that it came from a specific sender and was not altered in transit~\cite{dunjko2014,roberts2017measurement}. Both of these functions are critically important and constitute the foundation of modern cryptography, which is why developing quantum-resistant versions is a major focus of research.

For QIA, early protocols often relied on the principles of QKD as a foundational layer~\cite{bostrom2002,deng2004}. The logic was to use the tamper-evident nature of quantum communication to establish trust. However, it is well-established that QKD by itself does not solve the authentication problem; it secures a key exchange but cannot verify the identity of the participants at the outset~\cite{mayers1998,Scarani2009}. Later protocols have attempted to build more sophisticated QIA schemes, sometimes embedding authentication directly into other protocols or combining quantum techniques with classical methods like hash functions~\cite{chen2022}. Despite these advances, a common thread persists: the need for a trusted third party or, more fundamentally, a pre-existing authenticated classical channel to bootstrap the process~\cite{barnum2002}. This dependency is required to prevent man-in-the-middle attacks where an adversary could impersonate a legitimate party during the initial communication.

Quantum digital signature is to use quantum methods to sign on messages, either quantum one-way function ~\cite{gottesman2001} or relying on non-locality of Bell state, non-cloning theorem, offering a higher level of security than their classical counterparts. However, QDS protocols face significant practical limitations and dependencies. Early schemes explicitly required a trusted third-party, or ``arbitrator,'' to validate and authenticate the signed message, creating a central point of failure~\cite{lo2006,amiri2016}. More recent protocols, including Measurement-Device-Independent (MDI) schemes, have tried to remove this dependency but still explicitly require authenticated classical channels for coordination~\cite{roberts2017measurement,Pu2018mdi}. Furthermore, the very nature of a ``quantum public key'' makes it difficult to manage; the no-cloning theorem makes it physically impossible to freely copy and distribute the key, and many protocols require advanced, and still largely experimental, technology like quantum memory to store the fragile quantum states for verification~\cite{lvovsky2009,heshami2016}. Recent research aiming to remove the trusted party has had to introduce other strong assumptions, such as the requirement of an un-tamperable quantum channel for key transmission~\cite{dunjko2014,collins2014}.



In this work, we propose a QIA–QDS protocol that eliminates the need for pre-registration, trusted third parties, and pre-authenticated classical channels.
Specifically, in our scheme, each user’s private key is represented by a quantum state, while the corresponding set of local reduced density matrices functions as the public key. The digital signature is realized by encoding classical messages into the global quantum state before transmission, thereby ensuring strong guarantees of message integrity and authenticity. Crucially, the security foundation of our protocol lies in the QMA-completeness of the QMP, also known as the Consistency of Local Density Matrices (CLDM) problem~\cite{liu2006}. The QMA-complete problem analogous to classical NP-complete problems—remains computationally intractable even for quantum computers.

\section{Preliminaries: The Quantum Marginal Problem and Complexity}
The QMP is a fundamental question concerning the relationship between a whole quantum system and its parts ~\cite{Coleman1963,Klyachko2006}. Formally, given a set of $n$ particles indexed by the set $I=\{1,\ldots ,n\}$, a collection of index subsets $J_k\subset I$, and a corresponding set of density matrices $\rho_{J_k}$, the QMP asks for the conditions under which a global state $\rho_I$ exists such that for all $k$, $\operatorname{Tr}_{I\setminus J_k}(\rho_I)=\rho_{J_k}$. This problem is also known as the $N$-representability problem in quantum chemistry ~\cite{LiuChristandlVerstraete2007}. 

For cryptographic purposes, we focus on the associated decision problem, known as the CLDM problem ~\cite{KempeKitaevRegev06}. Throughout this paper, we will use QMP to refer to the general conceptual problem and CLDM to denote the precise computational problem that underpins our security proof.

\begin{definition}[CLDM problem~\cite{liu2006}]
Consider a system of $n$ qubits.  We are given a collection of local
density matrices $\rho_1,\ldots ,\rho_m$, where each $\rho_i$ acts on a
subset of qubits $C_i\subseteq\{1,\ldots ,n\}$.  Every matrix entry is
specified with $\operatorname{poly}(n)$ bits of precision.  We also have
$m\le\operatorname{poly}(n)$, and each subset satisfies
$\lvert C_i\rvert\le k$ for some constant $k$.

In addition, a real number $\beta$ is provided (again with
$\operatorname{poly}(n)$ bits of precision) such that
$\beta \ge 1/\operatorname{poly}(n)$.

The task is to distinguish between the following two cases:
\begin{itemize}
  \item[YES:] There exists an $n$-qubit state $\sigma$ such that, for
        all $i$,
        \[
          \bigl\lVert
            \operatorname{Tr}_{\{1,\ldots ,n\}\setminus C_i}(\sigma)
            -\rho_i
          \bigr\rVert_1 = 0 .
        \]
  \item[NO:] For every $n$-qubit state $\sigma$ there is some $i$ for
        which
        \[
          \bigl\lVert
            \operatorname{Tr}_{\{1,\ldots ,n\}\setminus C_i}(\sigma)
            -\rho_i
          \bigr\rVert_1 \ge \beta .
        \]
\end{itemize}
\end{definition}

The CLDM problem is known to be QMA-complete, indicating that it is as hard as the most difficult problems verifiable by quantum computation. To clarify this classification, we briefly introduce the QMA complexity class.
The complexity class QMA is the quantum analogue of the classical complexity class NP. In the QMA framework, an all-powerful but untrustworthy prover (Merlin) sends a quantum state, or "witness,"  $\vert\psi\rangle$ to a polynomial-time quantum verifier (Arthur). Arthur performs a verification circuit on the witness and outputs 'accept' or 'reject'. A problem is in QMA if it satisfies two conditions:

\textbf{Completeness:} For any YES instance of the problem, there exists a witness state $\vert\psi\rangle$ that causes Arthur to accept with high probability (e.g., $P(\mbox{accept}) > 2/3$).

\textbf{Soundness:} For any NO instance of the problem, every possible witness state causes Arthur to be rejected with high probability (i.e., $P(\mbox{accept}) < 1/3$). 

The gap between the acceptance probabilities for YES and NO instances is crucial and can be amplified by repeating the protocol. A problem is QMA-complete if it is in QMA and any other problem in QMA can be reduced to it in classical polynomial time. 

A remarkable result in quantum complexity theory is that the CLDM problem is QMA-complete~\cite{liu2006}.  The computational structure of QMA is what makes its complete problems suitable for cryptography. QMA problems are fundamentally promise problems. The verifier is guaranteed that the input instance belongs to one of two disjoint sets: YES instances, for which a "good" proof exists, or NO instances, for which no convincing proof can be constructed. For CLDM problem, this promise manifests as a gap: the given marginals are either highly consistent (a YES case) or any potential global state will be highly inconsistent with at least one of them (a NO case). A cryptographic protocol can exploit this gap to distinguish between legitimate and malicious behavior. An honest user's actions, by design, will correspond to a YES instance of the underlying problem. A successful forgery, as will be shown, would require the creation of a valid witness for a NO instance, an act deemed impossible by the soundness property of QMA. Thus, the promise gap inherent to the complexity class provides the necessary separation for cryptographic security.

\section*{The QMP-Based Cryptographic Protocol}

Our protocol consists of
three phases that together realize a quantum public-key scheme. In the key generation phase, Alice’s private key is a polynomial-depth circuit; her public key is the full set of k-qubit marginals of the circuit’s output state, checkable via local consistency. In the authentication phase, Bob challenges Alice with an arbitrary M-qubit subset; Alice then returns the corresponding fragment, and Bob verifies its marginals against the public key. In the digital signature phase, Alice encodes a message into a unitary generated from some publicly-known, message-dependent transformation. A;ice applies the unitary to the challenged fragment, and any verifier can invert the unitary and test the marginals. The scheme requires no pre-registration. Its security is based on the hardness of reconstructing a highly entangled state from sparse local data.

\subsection*{Key Generation}

To initiate the key generation process, Alice first selects a security parameter $\lambda$ and constructs a classical description of a quantum circuit $\text{Circuit}_A$ with depth $\mathrm{poly}(\lambda)$. Applying $\text{Circuit}_A$ to the fixed initial state $\ket{0}^{\otimes N}$ yields her private key, the $N$-qubit state $\rho_A$. Once the private key state is prepared, Alice computes all $k$-qubit marginals by performing state tomography on each overlapping subsystem of size $k$. The resulting set of classical density matrices forms her public key, which she publishes. This workflow starts with Alice using her private circuit to prepare the global state $\rho_{A}$. She then publishes all its k-qubit reduced states. Anyone can download these marginals and check that they fit together consistently. However, without knowing the exact circuit parameters in $\text{Circuit}_A$, rebuilding the full $N$-qubit state is believed to be computationally infeasible.

\textbf{Alice's Private Key (\textit{sk}\textsubscript{A}):}  
Alice's private key is a classical description of an efficient quantum circuit, $\text{Circuit}_A$. This circuit, when applied to a standard initial state like $\ket{0}^{\otimes N}$, prepares a specific, highly entangled $N$-qubit system. The choice of $\rho_A$ should be such that it is highly-entangled thus its global entangled structure would be destroyed or only partially exist locally. The generation of large, structured entangled states is an active area of experimental research. The classical description of an efficient quantum circuit, $\text{Circuit}_A$ is to be used to generate Alice's private key. $\text{Circuit}_A$ is subject to a security parameter $\lambda$. The depth of $\text{Circuit}_A$ is $\text{poly}(\lambda)$.

\textbf{Alice's Public Key (\textit{pk}\textsubscript{A}):}  
Alice defines a set of $k$-local overlapping subsystems, $\{C_1,C_2,\ldots ,C_{\binom{N}{k}}\}$, where $C_N^k$ is a combinatorial number and S is a collection of the indices of qubits in the $N$-qubit entangled system $\rho_A$. Alice then generates the marginal density matrix for each subsystem by state tomography. 

Her public key, $\textit{pk}_\text{A}$, is the set of classical descriptions of these $k$-local density matrices,
$\textit{pk}_\text{A}= \{\rho_{C_1},\rho_{C_2},\ldots ,\rho_{C_{\binom{N}{k}}}\}$, which she makes publicly available. By its construction, the set of local density matrices representing Alice's public key is perfectly consistent, with the state $\rho_A$ serving as the unique global-state witness to this consistency. The pseudocode of key generation is shown as Algorithm~\ref{alg:keygen}.

\begin{algorithm}[ht]
  \caption{\textsc{KeyGen()}}
  \label{alg:keygen}
  \begin{algorithmic}[1]
    \REQUIRE security parameter $\lambda$
    \ENSURE $(\mathit{sk}_A,\mathit{pk}_A)$
    \STATE Choose a circuit $\mathrm{Circ}_A$ of $\mathrm{poly}(\lambda)$ size to prepares an $N$-qubit state $\rho_A$.
    \STATE $sk_{A} = \rho_A$
    \STATE Define $C_N^k$ subsystems $\{C_1, C_2,\ldots ,C_{\binom{N}{k}}\}$.
    \FOR{$k = 1,\ldots, \binom{N}{k}$}
      \STATE Compute local density matrix $\rho_{C_k}= \operatorname{Tr}_{\{1,\ldots,n\} - C_k} (\rho_A)$
      \STATE Append local density matrix $\rho_{C_k}$ to public key $\mathit{pk}_A$
    \ENDFOR
    \RETURN $(\mathit{sk}_A,\mathit{pk}_A)$
  \end{algorithmic}
\end{algorithm}


\subsection*{Authentication 
}

We design a challenge-response protocol to prove Alice's identity with a verifier Bob. 

\textit{Challenge:}  
In this protocol, Bob first send a challenge to Alice by randomly selecting an $M$-qubits subsystem from $\{1,\ldots ,N\}$ qubits system of Alice, where $k<M<N$. He sends the classical description of all the indices of qubits and send this challenge to Alice. The state Bob asked for is denoted as $s_M$, which is a string of indices of corresponding qubits. 

\textit{Response:}  
After receiving the challenge string, Alice uses her private key $\text{Circuit}_A$ to prepare the state $\rho_A$. According to the challenge string $s_M$. She then sends the state $\rho_M$ to Bob as a response.

\textit{Verification:}  
Bob receives multiple copies of the subsystem state $\rho_M$ and performs quantum state tomography to reconstruct the corresponding $k$-qubit local density matrices, which we denote by $\rho_{C_k} \;=\;\operatorname{Tr}_{\{1,\ldots,M\}\setminus C_k}(\rho_M)$.
He then checks each reconstructed marginal against the corresponding public‐key marginal \(\rho_{C_k}\) by verifying
\[
\frac{1}{2}\bigl\|\rho_{C_k} - \rho_{C_k}\bigr\|_1 \;\le\;\epsilon,
\]
for every $C_k\subset s_M$ and $|C_k| = k$, where $\epsilon$ is a predetermined acceptance threshold. If every inequality holds,
Bob accepts that the responder is Alice, as only she can produce the global state $\rho_A$ from which these statistics arise.

\subsection*{Digital Signature 
}

\textit{Signing:}  
To sign a classical message $m$, Alice applies a publicly known, message-dependent, and efficiently invertible unitary transformation $U_m$ to the state asked by Bob, $\rho_M$. In many digital signature protocols, there is a preprocessing process: a plain-texted, arbitrary, unstructured \(x\) is first compressed through a publicly specified cryptographic hash function \(h\), producing the fixed-length message~\cite{preneel1994cryptographic}
\(
  m \;=\; h(x).
\)
The digest \(m\) is then a standardized message that enters the signature protocol. After transformation, the resulting quantum state,
$\sigma_m=U_m\rho_M$, then constitutes the quantum digital signature for the message $m$.  Alice prepares multiple identical copies of $\sigma_m$ and transmits them to the verifier. 

In the following paragraph, we first give a proper definition or to say, limitation on the message that is to be sent, and then give a proper definition for a message-dependent transformation $U_m$:


\begin{definition}[Classical messages]\label{def:message}
Let $\mathcal{L}$ be a finite, publicly agreed-upon alphabet, say language.
A message is a finite word
\[
  m \;=\; m_1 m_2 \dotsm m_{|m|},\qquad
  m_j \in \mathcal{L}\;(1\le j\le|m|).
\]
Typically $|m|\le \gamma$, where $\gamma$ is the
allowed message length.  A simple example is: with
$\mathcal{L}=\{0,1\}$ we recover ordinary binary strings. 
\end{definition}

\begin{definition}[Message-dependent unitary $U_m$]\label{def:Um}
A publicly known universal gate set is given to construct quantum circuit and give operation on qubits. Such gate set is
\[
  \mathcal{G}=\{G_1,G_2,\ldots,G_L\}.
\]
 For every
$i\in\{1,\dots,L\}$ and every $\ell\in\mathcal{L}$ we specify a unitary
$G_i^{(\ell)}$ via the public rule
\[
  G_i^{(\ell)} =
  \begin{cases}
    \mathbb{I}, & \text{if $\ell$ encodes ``skip''},\\[4pt]
    G_i,        & \text{if $\ell$ encodes a
                  non-parametric gate},\\[4pt]
    G_i(\theta_\ell), & \text{if $G_i$ is a rotation and
                         $\ell\!\mapsto\!\theta_\ell\in[0,2\pi)$}.
  \end{cases}
\]

\paragraph{Construction of $U_m$.}
Read $m=m_1\dotsm m_{|m|}$ from left to right and assign gates cyclically with
$i(j)=(j\bmod L)+1$, rightmost-first-ordered,
\begin{equation}
  U_m \;=\; \prod_{j=1}^{|m|} G_{\,i(j)}^{\bigl(m_j\bigr)}.
  \label{eq:Um}
\end{equation}

\paragraph{Properties.}
\begin{enumerate}[label=(\roman*), leftmargin=*, itemsep=2pt]
  \item Efficient invertibility.  
        $U_m^{-1}$ is obtained by reversing the product
        in~\eqref{eq:Um} and taking adjoints, so both
        $U_m$ and $U_m^{-1}$ have depth $O(|m|)$.

  \item Injectivity. 
        Different messages change at least one factor
        in~\eqref{eq:Um}; hence the map
        \(
          \mathcal{U}:\mathcal{L}^{\ast}\!\to\!\mathrm{U}(2^{n}),\;
          m\mapsto U_m
        \)
        is injective.

  \item Public computability. 
        Because the rule $(i,\ell)\mapsto G_i^{(\ell)}$ is public,
        both $\mathcal{U}$ and its inverse
        $\mathcal{U}^{-1}$ are efficiently computable.
\end{enumerate}
\end{definition}

\medskip
During the signing phase Alice applies $U_m$ from
Definition~\ref{def:Um} to the challenged subsystem $\rho_M$,
producing the signature state $\sigma_m = U_m\rho_M$. 
The pseudocode of signing a message by private key is shown as Algorithm~\ref{alg:sign}.

\begin{algorithm}[ht]
  \caption{\textsc{Sign}$(\textit{sk}_A,s_M,m)$}
  \label{alg:sign}
  \begin{algorithmic}[1]
    \REQUIRE private key $\mathit{sk}_A$, challenge $s_M$, message $m$
    \ENSURE multiple copies of $\sigma_m$
    \STATE Alice uses $\mathrm{Circ}_A$ to prepare $\rho_A$.
    \STATE Alice uses $s_M$ and $\rho_A$ to prepare $\rho_M$.
    \STATE Alice applies the public, message-dependent unitary $U_m$ to get $\sigma_m = U_m\rho_M$.
    \STATE Alice outputs multiple copies of $\sigma_m$.
  \end{algorithmic}
\end{algorithm}

\textit{Verification:}  
Any party in possession of Alice's public key $\textit{pk}_\text{A}$, the message $m$, and the signature copies $\sigma_m$ can perform verification. The verifier's goal is to confirm that the received state, when untransformed, has marginals consistent with Alice's public key. To do this, the verifier first applies the inverse transformation $U_m^{-1}$ to each copy of the signature, yielding the state $\sigma_m'=U_m^{-1}\sigma_m$. 
Verification then proceeds exactly as in the Authentication procedure, with each instance of $\rho_{M}$ replaced by $\sigma_m'$.
The methods for performing this check are detailed in the Section \emph{Security Analysis}. The pseudocode of verifying a signature is shown as Algorithm~\ref{alg:verify}.

\begin{algorithm}[ht]
  \caption{\textsc{Verify}$(\textit{pk}_A,m,\sigma_m)$}
  \label{alg:verify}
  \begin{algorithmic}[1]
    \REQUIRE public key $\textit{pk}_\text{A}= \{\rho_{C_1},\rho_{C_2},\ldots ,\rho_{C_{\binom{N}{k}}}\}$, message $m$, signature $\sigma_m$, acceptance threshold $\epsilon$
    \ENSURE \textsc{accept} or \textsc{reject}
    \STATE Verifier construct $U_m^{-1}$ by message $m$ and defined rule
    \STATE Verifier reduction the signature state to private key fragment: $\sigma_m' = U_m^{-1}\sigma_m$.
    \FOR{$k = 1,\ldots, \binom{M}{k}$}
      \STATE Verifier compute $\sigma_{C_k} \;=\;\operatorname{Tr}_{\{1,\ldots,M\}\setminus C_k}(\sigma_m')$
      \IF{$\frac{1}{2}\bigl\|\rho_{C_k} - \rho_{C_k}\bigr\|_1 > \epsilon$}
      \STATE \textbf{return} REJECT
      \ENDIF
    \ENDFOR
    \STATE \textbf{return} ACCEPT
  \end{algorithmic}
\end{algorithm}

\section*{Security Analysis}
\label{sec:security}
 After rigorously defining the digital signature and authentication protocol, we need to analyze the security of the proposed scheme. The analysis is to be conducted within the standard cryptographic framework of an adaptive chosen-message attack (CMA), which is extended to accommodate a quantum adversary. The adversary, Eve, is hereby modeled as a quantum polynomial-time algorithm. The security goal of this model is  to achieve existential unforgeability (EUF)~\cite{GoldwasserMicaliRivest1988}, which asserts that an adversary cannot produce a valid signature for any new message. 

The security is defined by the Existential Unforgeability under adaptive quantum Chosen Message Attack (EUF-qCMA) game~\cite{BonehZhandry2013}, where a challenger runs Algorithm \textsc{KeyGen()} to generate a key pair $(\mathit{sk}_A,\mathit{pk}_A)$ and provides the public key $\mathit{pk}_A$ to the adversary Eve. Eve is then given oracle access to a signing oracle, $\mathcal O_{\text{Sign}}$. She can adaptively make a polynomial number of queries, sending messages $m_1,\ldots ,m_q$ to the oracle. For each query $m_i$, the oracle uses $\mathit{sk}_A$ to compute the signature $\sigma_{m_i}$ and returns a set of identical copies to Eve. After the query phase, Eve outputs a message-signature pair $(m_E,\sigma_E)$, where $m_E$ is a message she did not query, i.e.\ $m_E\notin\{m_1,\ldots ,m_q\}$. Eve wins the game if the $\textsc{Verify}(\mathit{pk}_A,m_E,\sigma_E)$ procedure returns \textsc{Accept} with a probability that is non-negligible in the security parameter.

The signature scheme is considered secure if no quantum polynomial-time adversary can win the EUF-qCMA game with more than negligible probability. As the digital signature scheme is constructed on the identity authentication scheme, the proof of security of the authentication model is also given in This model is a quantum generalization of well-established classical security notions.

\subsection*{Proof of Unforgeability (Reduction to CLDM problem)}

The proof of unforgeability of our protocol proceeds by reduction. We demonstrate that if a quantum polynomial-time adversary E could successfully forge a signature, then we could construct another quantum polynomial-time algorithm F that uses E as a subroutine to solve the QMA-complete CLDM problem ~\citep{KempeKitaevRegev06,LiuChristandlVerstraete2007}. Since CLDM problem is believed to be intractable for quantum computers (that is, $\textsf{BQP} \neq \textsf{QMA}$), this implies that no such adversary E can exist.

\textbf{Theorem:} The QMP-based digital signature scheme is existentially unforgeable under adaptive chosen-message attacks, assuming $\textsf{BQP} \neq \textsf{QMA}$.

\textbf{Proof}:   
Assume, for the sake of contradiction, that there exists a quantum polynomial-time adversary E that wins the EUF-qCMA game with non-negligible probability $\delta$. We construct an algorithm F to solve a given instance of CLDM problem. Such algorithm F receives an instance of the CLDM problem, which consists of a set of $k$ local density matrices $\rho_{C_k}'$ and a promise that this set is either a YES instance (highly consistent) or a NO instance (highly inconsistent). F's task is to decide which is the case. Such algorithm F then initiates the EUF-qCMA game with the forger E and sets the public key for the game to be the CLDM instance it was given: $\mathit{pk}_E\leftarrow\rho_{C_k}'$. Using the oracle provided in the EUF-qCMA game, when E queries the signing oracle for a message $m_i$, B is faced with a challenge: it cannot generate the signature because it does not know the global state $\rho_E$ corresponding to the public key (and for a NO instance, no such state exists). However, the reduction cleverly avoids this issue. The security proof does not require B to answer the queries correctly. The existence of a successful forger is assumed regardless of how oracle queries are handled. 

After its queries, the adversary E outputs its forgery: a pair $(m_E,\sigma_E)$ for a new message $m_E$. By our initial assumption, this forgery must pass the verification check with non-negligible probability. Algorithm F takes the forged quantum state $\sigma_E$ and applies the publicly known inverse unitary $U_{m_E}^{-1}$ to obtain the state $\sigma_m' = U_{m_E}^{-1}\sigma_E$. According to the \textsc{Verify} algorithm, for the signature to be valid, the marginals of $\sigma_m'$ must be consistent with the public key $\rho_{C_k}'$. This means the state $\sigma_m'$ is a quantum witness that satisfies the consistency conditions of the original CLDM problem instance provided to F.  F can now use $\sigma_m'$ to solve the CLDM problem. It submits the state $\sigma_m'$ as a witness to a standard QMA verifier for CLDM. If the original CLDM instance was a NO instance, the soundness property of QMA guarantees that no quantum state can serve as a convincing witness. Therefore, if E produces a forgery, the CLDM instance given to F could not have been a NO instance. If the original CLDM instance was a YES instance, a valid witness exists, and the forger E might succeed.

By observing whether the forger succeeds in producing a valid witness, F can distinguish between YES and NO instances of CLDM problem . E successful forgery by E implies the instance is YES. The absence of a successful forgery (over many runs) implies the instance is NO. This allows F to solve CLDM problem with a non-negligible advantage, which contradicts the assumption that CLDM problem is QMA-complete.

Therefore, the initial assumption must be false: no such polynomial-time quantum adversary E can exist. The signature scheme is secure.

\subsection*{Authentication Security}

The security of the authentication protocol follows a similar logic. An imposter, Eve, attempting to respond to Bob's challenges would need to produce quantum states whose measurement statistics on a subsystem $C_k$ match those of the public marginal $\rho_{C_k}$. To do this successfully for arbitrary challenges across all subsystems, Eve would effectively need to possess a global state consistent with the entire set $\rho_{C_k}$. The ability to generate such a state on demand is equivalent to solving the CLDM problem.

\subsection*{Non-Repudiation and Transferability}

The protocol provides essential properties for a digital signature scheme.

\textbf{Non-Repudiation}: Alice cannot deny having signed a message $m$ if a valid signature $\sigma_m$ exists. The verification process is public and relies only on publicly available information ($\mathit{pk}_A$, $m$). If \textsc{Verify} accepts, it is a mathematical proof that the provided state has the correct properties relative to the public key. The link between Alice's identity and her public key is a prerequisite for any public-key system and is typically handled by a public ledger or directory. 

\textbf{Transferability}: A recipient, Bob, who has received and verified a signature $(m,\sigma_m)$, can forward these to a third party, Victor. Victor can independently perform the same \textsc{Verify} procedure using Alice's public key to convince himself of the signature's validity. This transferability is a direct consequence of the public nature of the verification algorithm.

\subsection*{Necessity of the $k<M<N$ restriction} 
A fundamental design choice in both the authentication and the
signature protocol is that Bob’s challenge never asks Alice to reveal
her entire $N$-qubit private state.  Instead, he selects a random
subsystem of size $M$ with
\(
  k < M < N .
\) We justify this restriction with the following lemma.

\begin{lemma}[State\,\&\,Key-extraction attack]\label{lem:full-leak}
Suppose an adversarial verifier is allowed, in a single session, to
demand the full $N$-qubit state
$\rho_A$ that serves as Alice’s private key.
Then after that session the verifier can, with overwhelming
probability, impersonate Alice in all future executions of the
protocol and forge signatures for arbitrary messages.
\end{lemma}

\begin{proof}
Once the adversary receives $\rho_A$, it can store the state in a
quantum memory and reuse it indefinitely; no inverse transformation or
measurement is required.  In the authentication protocol, responding to
any future challenge merely means measuring the appropriate subsystem
of that stored state; hence the adversary’s success probability is~$1$.

For the digital signature protocol, recall that Alice signs a message
$m$ by applying the public efficient invertible unitary $U_m$ to
$\rho_A$, producing $\sigma_m=U_m\rho_M$, where $\rho_M$ is generated from $\rho_A$ and new challenger $s_M$.  Because the
adversary now possesses $\rho_A$, it can reproduce
exactly the same procedure:
\[
  \sigma_m^{\text{fake}}
  \;=\;
  U_m\rho_M.
\]
Verification applies $U_m^{-1}$ and checks the marginals of the
resulting state against those published in the public key; the forged
state passes with certainty.  Hence existential unforgeability is
utterly broken once the full $N$-qubit key leaves Alice’s
laboratory.
\end{proof}

Lemma~\ref{lem:full-leak} shows that exposing the entire state would
collapse security to the trivial level: Bob (or any malicious
verifier) could record $\rho_A$ and become a perfect clone of
Alice.  By limiting each challenge to an $M$-qubit slice,
with $M$ strictly less than $N$, we prevent any single verifier from
obtaining enough information to reconstruct the global state, guaranteed by the QMA-completeness of the CLDM problem. Moreover, as the locations of the $M$ qubits are chosen randomly each time,
collecting the full state's all possible subsystem of $M$ qubits would require $\binom{N}{M}$ protocol
runs, during which Alice would notice the abnormal requests and related key leakage.  This subsampling strategy is therefore essential to preserve
both impersonation resistance and signature unforgeability while still
allowing efficient verification.


\subsection*{Verification Efficiency and Practical Considerations}

The practical viability of this protocol hinges on the efficiency of
the \textsc{Verify} algorithm. The core task is to check the consistency
condition $\lVert \operatorname{Tr}_{\{1,2,\ldots,\binom{N}{M}\}- C_k}{\sigma_m'}\! -
\rho_{C_k}\rVert \le \epsilon$ for each subsystem $C_k$. We analyze two
prominent quantum procedures for this task --- to use partial Quantum State Tomography (pQST) technique.

\paragraph*{Verification via Partial Quantum State Tomography (pQST)}
In this approach, the verifier reconstructs a classical description of
the marginals of the received state and compares them to the public
key. 

For each subsystem $C_k$, the verifier uses the
provided copies of
$\sigma_m' = U_{m}^{-1}\ket{\sigma_{m}}$ to perform quantum state
tomography on that $k$-qubit subsystem. This yields an estimate of the local subsystem ${\rho}_{C_k'}$. The verifier then classically computes a distance metric, such as the trace distance, between the reconstructed
local density matrix ${\rho}_{C_k'}$ and the public key subsystem
$\rho_{C_k}$.

\textit{Resource Consumption:} While full tomography of an $N$-qubit state is infeasible, scaling exponentially with $N$, pQST is only performed on
small, $k$-qubit subsystems. The number of state copies required to
achieve a precision $\epsilon$ for a $d$-dimensional system
($d = 2^{k}$) scales as
$\mathcal{O}(d^{2}/\epsilon^{2}) = \mathcal{O}(4^{k}/\epsilon^{2})$
~\citep{ParisRehacek2004}. For a small and fixed subsystem size $k$,
this is efficient. The total cost is polynomial in the number of
local subsystems.

In the realistic implementation, the quantum states used in our protocol will be subject to decoherence due to environmental interactions and operational imperfections. Such noise will affect both Alice's preparation of her private state and the transmission of the signature state to the verifier. As a result, even an honest signature will not pass a perfect verification check.

The protocol must therefore incorporate an error threshold. The verifier will accept a signature if the measured inconsistency is below a threshold. This threshold must be carefully calibrated: it must be large enough to tolerate the expected level of natural decoherence but small enough to reliably detect malicious modifications that would constitute a forgery. Ultimately, for the protocol to be truly scalable and secure over long distances or long computational times, it may be implemented using logical qubits protected by a Quantum Error Correction (QEC) code. QEC schemes encode the information of a single logical qubit across many physical qubits, allowing for the detection and correction of errors. The security analysis presented in this paper assumes ideal, error-free qubits and serves as the theoretical foundation upon which a fault-tolerant version of the protocol can be built.

\section{Discussion}
\label{sec:discussion}
This work has introduced a novel framework for public-key quantum cryptography based on the computational hardness of the quantum marginal problem. The resulting authentication and digital signature protocol is, to our knowledge, the first to leverage the QMA-completeness of a natural physical problem to achieve security.

The protocol's principal advantage is its self-contained and decentralized nature. It successfully establishes a true public-key system-without any reliance on trusted third parties, pre-shared secrets between users, or pre-authenticated classical channels for the distribution of public keys. This represents a significant step toward building scalable quantum networks where trust can be established dynamically and securely based on the laws of quantum mechanics and computational complexity. The security is proven to be robust, with existential unforgeability against adaptive chosen-message attacks by quantum adversaries reducible to the intractability of the CLDM problem.

Based on this new formalism of quantum cryptography, we hereby propose several promising avenues for future research based on our protocol. First, the protocol can be optimized by exploring different families of global states $\rho$ and different configurations of overlapping subsystems $\{C_k\}$ to find the ideal balance between the strength of the security assumption and the resource costs of key generation and verification. Second, a small-scale proof-of-principle experiment on current noisy intermediate-scale quantum (NISQ) hardware is a direct next step. Such an experiment could involve generating appropriate quantum state as private key 
, distributing its local subsystems as a public key, and performing the verification steps to demonstrate the protocol's core mechanics, even in the presence of noise. Finally, the core methodology—using the hardness of a QMA-complete problem as a cryptographic primitive—could be applied to other quantum-computationally hard problems, including estimating the ground-state energy of specific local Hamiltonians or verifying properties of quantum circuits, potentially leading to new cryptographic functionalities with different security and efficiency profiles.

\normalem
\bibliography{references}
\end{document}